\documentclass[letterpaper, 10 pt, conference]{ieeeconf}  
\IEEEoverridecommandlockouts                             
\usepackage{paralist}
\usepackage{comment}
\usepackage{amsfonts}
\usepackage{graphics} 
\usepackage{epsfig} 
\usepackage{mathptmx} 
\usepackage{times} 
\usepackage{amsmath} 
\usepackage{amssymb}  
\usepackage{physics}
\usepackage{mdwmath}
\usepackage{mdwtab}
\usepackage{color}
\usepackage[hidelinks]{hyperref}
\usepackage{algpseudocode}
\usepackage{algorithm}
\usepackage{caption}
\usepackage{subcaption}

\newtheorem{prop}{Proposition}
\newtheorem{remark}{Remark}

\title{\LARGE \bf
Quantum Pontryagin Neural Networks in Gamkrelidze Form Subjected to the Purity of Quantum Channels}

\author{Nahid Binandeh Dehaghani, A. Pedro Aguiar, Rafal Wisniewski
\thanks{N. Dehaghani and A. Aguiar are with the Research Center for Systems and Technologies (SYSTEC), Electrical and Computer Engineering Department, FEUP - Faculty of Engineering, University of Porto, Rua Dr. Roberto Frias sn, i219, 4200-465 Porto, Portugal
        {\tt\small \{nahid,pedro.aguiar\}@fe.up.pt}}%
\thanks{R. Wisniewski is with Department of Electronic Systems, Aalborg University, Fredrik Bajers vej 7c, DK-9220 Aalborg, Denmark
        {\tt\small raf@es.aau.dk}}%
\thanks{The authors acknowledge the support of FCT for the grant 2021.07608.BD, the ARISE Associated Laboratory, Ref. LA/P/0112/2020, and the R$\&$D Unit SYSTEC-Base, Ref. UIDB/00147/2020, and Programmatic, Ref. UIDP/00147/2020 funds, and also the support of projects SNAP, Ref. NORTE-01-0145-FEDER-000085, and  RELIABLE (PTDC/EEI-AUT/3522/2020) funded by national funds through FCT/MCTES. The work has been done in the honor and memory of Professor Fernando Lobo Pereira.}
}

\begin{document}

\maketitle
\thispagestyle{empty}
\pagestyle{empty}

\begin{abstract}
We investigate a time and energy minimization optimal control problem for open quantum systems, whose dynamics is governed through the Lindblad (or Gorini-Kossakowski-Sudarshan-Lindblad) master equation. The dissipation is Markovian time-independent, and the control is governed by the Hamiltonian of a quantum-mechanical system. We are specifically interested to study the purity in a dissipative system constrained by state and control inputs. The idea for solving this problem is by the combination of two following techniques. We deal with the state constraints through Gamkrelidze revisited method, while handling control constraints through the idea of saturation functions and system extensions. This is the first time that quantum purity conservation is formulated in such framework. We obtain the necessary conditions of optimality through the Pontryagin Minimum Principle. Finally, the resulted boundary value problem is solved by a Physics-Informed Neural Network (PINN) approach. The exploited Pontryagin PINN technique is also new in quantum control context. We show that these PINNs play an effective role in learning optimal control actions.
\end{abstract}
\section{INTRODUCTION}
During the last few decades, control of quantum systems has been extensively investigated. Recently, dissipative quantum systems have attracted much more attention in comparison with the quantum conservative systems, which are normally described by the time-dependent Schrödinger equation, for pure states captured by a wave function, or Liouville–von Neumann equation (LVNE), for both pure and mixed states described by the density operator
rather than the wavefunction. The dynamics of open quantum systems is governed through the extension of LVNE by considering dissipative processes, known as the so-called Lindblad master equation. Due to the vast application of dissipative systems in various areas of knowledge such as quantum physics, chemistry and also quantum computation and information processing, the control of such systems is of high interest. Control of open quantum systems ranges from the control of dissipation resulted from molecular collisions to the fabrication of quantum computers. 

Amongst quantum control tools, optimal control theory (OCT) has been applied to an increasingly extensive number of problems. Although OCT has been very well exploited for closed quantum systems, the control of dissipative quantum systems is still open to new ideas. Important examples of what may be thought of as the optimal control of dissipative quantum systems are the analytic solution for cooling the three-level $\Lambda$ systems, \cite{sklarz2004optimal}, optimal control of quantum purity for dissipative two-Level open quantum systems, \cite{clark2017optimal}, and quantum optimal control problem with state constrained preserving coherence, \cite{dehaghani2022a}. In optimal control theory and particularly for state-to-state population transition problem, it is possible to find the best achievable control to take a dynamical system from one initial state to predefined target state, especially in the presence of state or control constraints by exploiting Pontryagin's Maximum Principle. In \cite{dehaghani2022a}, we addressed the problem of maximizing fidelity in a quantum state transition process such that the coherence is preserved within some bounds. We formulated the quantum optimal control problem in the presence of state constraints and obtained optimality conditions of Pontryagin’s Maximum Principle in Gamkrelidze’s form for the first time in quantum control context. 

In this work, we consider an optimal control problem of quantum dissipative dynamics, in which the dissipation is Markovian time-independent. This means that the memory effects are neglected and the dynamics of the quantum system under study is only dependent on the current state not the past history. Therefore, we can describe the density operator evolution by the Lindblad master equation. In such dynamics, there is an interaction of the controllable part of the system, the Hamiltonian (or conservative) term, and the uncontrollable part, the non-Hamiltonian (or dissipative) term. Under most conditions, dissipation leads to an increase in entropy (or a decrease in purity) of the system. However, proposing an strategy to control the Hamiltonian term of the system evolution with the intent that the non-Hamiltonian term causes an increase in purity rather than decrease is a matter of debate. In this regard, we consider two types of constraint, namely state and control constraints, and preserve the quantum purity in a state transition problem. We then exploit the newly developed Pontryagin Physics-Informed Neural Networks (PINN) method derived from the Theory of Functional Connection (TFC), \cite{d2021pontryagin, johnston2021theory}, to learn the optimal control actions and solve the Boundary Value Problem (BVP) associated with the necessary optimality conditions resulted from the indirect application of Pontryagin Minimum Principle (PMP) in Gamkrelidze form.
\subsection{Main Contributions:}
We investigate the optimal control problem of minimum-time and minimum-energy for dissipative quantum systems interacting with the environment, whose dynamics is governed by
the Lindblad equation. Our contributions are itemized as follows
\begin{itemize}
    \item We consider the purity of quantum channels as state constraints, such that the purity is preserved within some bounds. The technique to deal with this problem leans upon Gamkrelidze revisited method, which is novel in quantum control context. (We have shown the feasibility of this method for coherence preservation in a different framework in one of our recent works, \cite{dehaghani2022a}.)
    \item We also consider control constraints, such that the value of control function remains in a defined interval. We handle control constraints through the idea of saturation functions and system extensions. This technique is also new in quantum control problems.
\item We derive the necessary optimality PMP conditions, show under which the necessary conditions are also sufficient conditions for (local) optimality, and finally solve the boundary value problem resulted from the PMP. To do so, we use and adapt a recently developed neural network approach known as the Pontryagin neural networks. The application of this method is new in quantum control problems, where we solve the state constrained PMP in Gamkrelidze form by means of Pontryagin Neural Networks (PoNNs).
\end{itemize}
Overall, the combination of these two techniques is new for an optimal control problem. Moreover, this is for the first time that purity preservation is formulated and solved in such framework. We also study the fidelity criteria in terms of purity, and see the effects of state constraints of purity preservation on fidelity. 

The structure of the work is as follows: In Section II, we introduce the Lindblad master equation and provide a linear isomorphism such that the transformed system is linear in the state and real. Section III formulates the optimal control problem of minimum time and energy, and in the next section we present the Pontryagin's Minimum Principle in Gamkrelidze form with saturation functions. We solve the resulted BVP through a physics-informed neural network approach, explained in section V. The simulation results have been shown for the quantum state transfer problem in a two-level system in section VI. The paper ends with conclusion and an overview on prospective research challenges.

\subsection{Notation.} For a general continuous-time trajectory $x$, the term $x(t)$ indicates the trajectory assessed at a specific time $t$. For writing the transpose of a matrix (or a vector) we use the superscript $T$, and we use $\dagger$ to show the conjugate transpose of a matrix (or vector). To denote the wave functions as vectors, we use the Dirac notation such that $\left| \psi  \right\rangle =\sum\limits_{k=1}^{n}{{{\alpha }_{k}}\left| {{{\hat{\psi }}}_{k}} \right\rangle }$, where $\vert \psi \rangle $ indicates a state vector, ${\alpha }_{k}$ are the complex-valued expansion coefficients, and $\vert {\hat \psi}_k \rangle $ are basis vectors that are fixed. The notation bra is defined such that $\left \langle \psi  \right |= \left| \psi  \right\rangle^\dagger$. In addition, the notation $\left|  \rho  \left. \right\rangle\right\rangle$ indicates the vectorized form of the density operator in the Fock-Liouville space, and the vectorization operator is shown by $vec$. For writing partial differential equations, we denote partial derivatives using subscripts. In the general situation where $f$ denotes a function of $n$ variables including $x$, then $f_x$ denotes the partial derivative relative to the $x$ input. The sign $\otimes$ indicates the tensor product. Finally, the notation $[\cdot,\cdot]$ represent a commutator and $\{\cdot,\cdot\}$ is a Poisson bracket. Throughout the paper, the imaginary unit is $i = \sqrt{-1}$. 

\section{The Lindbladian dynamics equation}
The general mathematical tool to describe our knowledge of the state of an n-level quantum system is through the density operator $\rho$, which is a Hermitian positive semi-definite operator of trace one acting on the Hilbert space $\mathbb{H}$ of the system. In quantum mechanics, the evolution along time $t$ of density operator through the Lindblad (or Gorini-Kossakowski-Sudarshan-Lindblad) master equation represents the most extensive generator of Markovian dynamics which takes the form, \cite{dehaghani2022quantum},
\begin{equation}\label{1}
\!\!\! \!\!\!\dot{\rho }\left( t \right)=-{i}\left[ H\left( u(t) \right),\rho \left( t \right) \right]+\sum\limits_{k}{{{\gamma }_{k}}\left[ {{L}_{k}}\rho L_{k}^{\dagger }-\frac{1}{2}\left\{ L_{k}^{\dagger }{{L}_{k}},\rho  \right\} \right]}\!
\end{equation}
in which the first term represents the unitary evolution of the quantum system, where $H\left( u(t) \right)$ is a Hermitian operator called the quantum-mechanical Hamiltonian defined as 
\begin{equation*}
H\left( u(t) \right)=
\underbrace{diag(E_1,E_2,\cdots,E_n)}_{{H}_{d}}+\underbrace{\sum\limits_{l=1}^m{{{u}_{l}}(t)}{{H}_{l}} }_{{H}_{C}}
\end{equation*}
where $H_d$ is diagonal implying that the basis corresponds to the eigenvectors, and $E_i$ represents a real number concerning the energy level. The control $u_l(t)\in \mathbb{R}$ demonstrates a set of external functions coupled to the quantum system via time independent interaction Hamiltonians ${{H}_{l}}$. The second term of \eqref{1} represents the dissipative part of the state evolution, where $L_{k}$ indicates the sequence of arbitrary Lindblad operators, and ${\gamma }_{k}\ge 0$ is the damping rate.
Master equations can be bothersome due to the commutation term and the Poisson bracket. Since convex combination gives the preservation of the trace and positive definiteness, it is possible to create a Hilbert space of density operators via defining a scalar product. The next result provides a linear space of density operators, called Fock-Liouville space, through which we present a solution of the master equation via vectorization such that the resulting Liouvillian superoperator is governed by a linear system.
\begin{prop}
Consider the Lindbladian dynamical system \eqref{1}. There exists a linear isomorphism, coordinate transformation, such that the transformed system is linear in the states and real.
\end{prop}
\begin{proof} The result is proved by providing a coordinate transformation given by the following composition of 3 linear isomorphisms (whose notation is defined in the sequel):
    \begin{inparaenum}[i)]
    \item  \label{itm:first}  $\mathcal{L}:~\rho \mapsto \mathcal{L}   \rho$, \quad
    \item \label{itm:second}  $\left| \cdot  \left. \right\rangle\right\rangle:~\rho \mapsto \left| \rho  \left. \right\rangle\right\rangle$,\\
    \item \label{itm:third} $\tilde \cdot:~ v \mapsto \tilde v := 
\begin{bmatrix} \operatorname{Re}(v) & \operatorname{Im}(v)\\ \end{bmatrix}^T$\\
\end{inparaenum} 
and finally $\tilde{\mathcal{L}} = \tilde \cdot \circ \left| \cdot  \left. \right\rangle\right\rangle \circ \mathcal{L} \circ \left| \cdot  \left. \right\rangle\right\rangle^{-1} \circ \tilde \cdot ^{-1}$.
To obtain \ref{itm:first}, we apply the Choi-Jamiolkowski isomorphism for vectorization through the mapping $\left| \cdot  \left. \right\rangle\right\rangle:~ \left| i \right \rangle \left\langle  j \right|\mapsto \left| j \right\rangle \otimes \left| i \right\rangle$, to \eqref{1} obtaining the Liouville superoperator acting on the Hilbert space of density operator as
$\left| \dot{\rho }  \left. \right\rangle\right\rangle=\mathcal{L} \left|  \rho  \left. \right\rangle\right\rangle$ where
\begin{equation}\label{5}
\begin{aligned}
 \mathcal{L}=&-i\left( I\otimes H-{{H}^{T}}\otimes I \right) \\ 
 & +\sum\limits_{k}{{{\gamma }_{k}}\left[ L_{k}^{*}\otimes {{L}_{k}}-\frac{1}{2}I\otimes L_{k}^{\dagger }{{L}_{k}}-\frac{1}{2}{{\left( L_{k}^{\dagger }{{L}_{k}} \right)}^{T}}\otimes I \right]} \\ 
\end{aligned}
\end{equation}
in which $I$ is the identity matrix. To show \ref{itm:second}, note that, for an arbitrary density operator $\rho =\sum\limits_{i,j}{{{\rho }_{i,j}}\left| i \right\rangle \left\langle  j \right|}$, its vectorized form is given by $   \left|  \rho  \left. \right\rangle\right\rangle=\sum\limits_{i,j}{{{\rho }_{i,j}}\left| j \right\rangle \otimes }\left| i \right\rangle$. Finally to get \ref{itm:third}, we implement the system state and superoperator as 
\begin{equation*}
\left|\left.\Tilde{\rho} \right\rangle\right\rangle =\left[ \begin{matrix}
   \operatorname{Re}\left(\left|\left.{\rho} \right\rangle\right\rangle \right)  \\
   \operatorname{Im}\left(\left|\left.{\rho} \right\rangle\right\rangle  \right)  \\
\end{matrix} \right], \quad \tilde{\mathcal{L}}:=\left( \begin{matrix}
   \operatorname{Re}\left( \mathcal{L}  \right) & -\operatorname{Im}\left( \mathcal{L} \right)  \\
   \operatorname{Im}\left( \mathcal{L} \right) & \operatorname{Re}\left( \mathcal{L} \right)  \\
\end{matrix} \right)
\end{equation*}
so that we obtain 
\vspace{-0.1cm}
\begin{equation}\label{8}
   \left| \dot{\Tilde{\rho} }  \left. \right\rangle\right\rangle=\tilde{\mathcal{L}} \left|  {\Tilde{\rho}}  \left. \right\rangle\right\rangle
   \vspace{-0.3cm}
\end{equation}
\end{proof}
We consider \eqref{8} as the system dynamics in the rest of paper.

\section{OPTIMAL CONTROL FORMULATION}
This section describes the problem formulation. To this end, we have first to describe the state constraint that arises from the concept of quantum purity.
\subsection{Quantum purity}
Quantum purity is a fundamental property of a quantum state. For pure states, the purity $P=1$, while $P<1$ shows that the quantum state is mixed. In dissipative quantum systems, a state may be initialized as pure, i.e., $\rho =\left| \psi  \right\rangle \left\langle  \psi  \right|$, and then due to the interaction with the environment and through a channel $\chi$ is decohered and mapped to a mixed state. A quantum channel $\chi$ over the space $\mathbb{H}$ represents a completely positive trace-preserving quantum map, i.e., $\chi \in \textit{CPTP} (\mathbb{H})$. Hence, the purity of the channel $\chi$ 
is considered as just the purity of the state $\rho$. Therefore, one can write 
\begin{equation*}
P(\rho)=tr\left( \chi {{\left( \rho  \right)}^{2}} \right)=tr\left( {{\rho }^{2}} \right)={{\left( vec\left( {{\rho }^{\dagger }} \right) \right)}^{\dagger }}vec\left( \rho  \right)
\end{equation*}
which leads to $\tilde{P}(\rho)=\left\langle\left\langle {{{\tilde{\rho }}}^{\dagger }}\left| {\tilde{\rho }} \right. \right\rangle\right\rangle$.
Preservation or maximization of the purity of a state transmitted through a quantum channel, i.e., a dissipative quantum system, is an important objective in quantum information processing. To do so, several decoherence-reduction techniques, such as quantum error correcting codes, and decoherence-free subspaces have been developed. In this work, we keep quantum purity above some predefined level by imposing the constraint as
\begin{equation*}
 \alpha {{P}_{0}}\le \Tilde{P}(\rho)\le {{P}_{0}}, \quad 0<\alpha <1
\end{equation*}
where $P_0=P(\rho_0)$ is the purity of the initial state. 
\subsection{Problem Formulation}
Quantum operations such as quantum state transition need to be done in the shortest possible time. However, due to the inverse relation between control time and amplitude, a fast operation may cause a very large control amplitude, which is practically impossible. The methods to design the quantum optimal controller vary according to the choice of the cost functional, the construction of the Pontryagin-Hamiltonian function, and the computation scheme to solve the PMP optimality conditions. Here, we deal with a time- and energy minimization state constrained optimal control problem ($P$) with bounded control aiming to transfer the initial state $\left|\left.{\Tilde{\rho}} (t_0)\right\rangle\right\rangle ={{\rho }_{0}}$ to a desired target state $\left|\left.{\Tilde{\rho}} (t_f)\right\rangle\right\rangle ={{\rho }_{f}}$. The problem casts as the following:
\begin{equation*}
(P)\left\{ \begin{aligned}
  &\min_{u,t_f}  \left\{J = \Gamma\, {{t}_{f}}+\eta\int_{{{t}_{0}}}^{t_f}{{{u}^{2}}\left( t \right)\,dt}\right\}\\ 
  &\text {subject to }  \\ 
  &\left|\left.\dot{\Tilde{\rho}}(t)\right\rangle\right\rangle=\tilde{\mathcal{L}} \left|  {\Tilde{\rho}}  \left. \right\rangle\right\rangle  \text{ a.e. } t\in [t_0,t_f] \\
  & \left|\left.{\Tilde{\rho}} (t_0)\right\rangle\right\rangle ={{\rho }_{0}}\in {{\mathbb{R}}^{4n}}\\ 
  & u\left( t \right)\in \mathcal{U}:=\left\{ u\in {{L}_{\infty }}:u\left( t \right)\in \Omega \subset {{\mathbb{R}}} \right\}\\
  & \Omega=[{{u}_{\min }},{{u}_{\max}}]  \text{ a.e. } t\in [t_0,t_f]\\   
  & h\left( \left|\left.\Tilde{\rho} (t)\right\rangle\right\rangle \right)\le 0 \text{ for all } t\in [t_0,t_f] \\
\end{aligned} \right.
\end{equation*}
where $\Gamma$ and $\eta$ in the performance index $J$ are positive coefficients, and ${{t}_{f}}$ shows the free final time to be optimized. The second term of $J$ is a common choice for the cost functional in molecular control, which measures the energy of the control field in the interval $[t_0,t_f]$. The control is represented as a measurable bounded function. 
The inequality $ h\left( \left|\left.\Tilde{\rho} (t)\right\rangle\right\rangle \right)\le 0$ defines the state constraints for the density operator - see the specific example in \eqref{18}. 
In this setup, we assume that all the sets are Lebesque measurable and the functions are Lebesgue measurable and Lebesgue integrable. The goal is to obtain a pair $\left( {{\rho }^{*}},{{u}^{*}} \right)$ which is optimal in the sense that the value of cost functional is the minimum over the set of all feasible solutions. 
\begin{remark}
It is important to assert the existence of a solution to Problem (P) in the class of measurable controls. Following the Filippov’s theorem, \cite{Filippov}, since the right-hand side of the dynamical system is linear with respect to the control, and the set of control values is convex and compact, then one can conclude that there exists a feasible control process to this problem. 
\end{remark}

\section{Pontryagin’s Minimum Principle in Gamkrelidze’s form with saturation functions}
To deal with the indicated problem, one can identify two Lagrangian multipliers: $\mu$, and $\lambda$, where
\begin{itemize}
  \item $\mu$ is bounded variation, non-increasing $\mu:\left[ t_0,t_f \right]\to {{\mathbb{R}}^{2}}$, such that $\mu(t)$ is constant on the time interval in which the state constraint is inactive. 
   \item $\left| \left. \lambda  \right \rangle \right\rangle:\left[ t_0,t_f \right]\to {{\mathbb{R}}^{4n}} $ is the time-varying Lagrange multiplier vector, whose elements are called the costates of the system.
\end{itemize}
We handle control constraints with saturation functions, \cite{graichen2010handling}, such that the indicated inequality-constraint for control is transformed into a new equality-constraint. To do so, we define a new unconstrained control variable $\nu(t)$, and substitute the control constraint with a smooth and monotonically increasing saturation function $\phi: \mathbb{R} \to (u_{\min},u_{\max})$ such that 
\begin{equation*}
   {\phi }\left( {\nu} \right)=u_{max}-\frac{u_{max}-u_{min}}{1+\exp \left( s{{\nu}} \right)}\quad\text{with}\quad s=\frac{c}{u_{\max}-u_{\min}}
\end{equation*}
in which $c>0$ is a constant parameter, useful for modifying the slope of $\phi(\nu)=0$. The advantage of using a saturation
function is that it is defined within the range of $\Omega$, and asymptotically approaches the saturation limits for $\nu \to \pm \infty$.
The next steps are the following:
\begin{itemize}
 \item We add a regularization term to the cost functional $J$ via a regularization parameter $\alpha$, and define the new cost functional as $\Tilde{J}=J+\alpha \int_{t_0}^{t_f}\nu^{2}(t)dt$, and solve the optimal control problem successively by decreasing $\alpha_{k}$. We use the result that if $u_{k+1}$ and $u_{k}$ are the optimal control inputs for $\alpha_{k+1}<\alpha_{k}$, then with $\underset{k\to \infty }{\mathop{\lim }}\,{{\alpha }_{k}}=0$, $\Tilde{J}(u_{k},{\alpha }_{k})$ converges to a non-increasing optimal cost, \cite{graichen2010handling}. In other words, by bringing $\alpha$ closer to 0, we approach to the original problem.
\item We consider an additional optimality condition for the new control variable by minimizing the Pontryagin-Hamilton function with respect to $\nu$. Moreover, 
we need to consider the constraint equation
 \vspace{-0.3cm}
      \begin{equation}\label{16}
        u(t)-\phi(\nu)=0
        \vspace{-0.2cm}
      \end{equation}
      for the boundary value problem.  
\item We introduce a multiplier $\beta:\left[ {{t}_{0}},{{t}_{f}} \right]\to \mathbb{R}$ to take the  equality constraint into account.      
\end{itemize}
In this new setup, we will now obtain the optimality conditions.
To this end, we first construct the Pontryagin Hamiltonian $\mathcal{H}$ defined for all $t\in \left[ {{t}_{0}},{{t}_{f}} \right]$ by
\begin{equation}\label{20}
\begin{aligned}
&\!\!\!{\mathcal{H}}(\rho,\lambda,u, \nu, \delta, \alpha, \beta ,t)\!=\!(\left| \left. \lambda \left( {t} \right) \right\rangle\right\rangle-2\delta(t)\left|\left.\Tilde{\rho}(t)\right\rangle\right\rangle)^T\Tilde{\mathcal{L}}\left|\left.\Tilde{\rho}(t)\right\rangle\right\rangle\\
 &\quad+\eta {{u}^{2}}\left( t \right)+\alpha{{\nu^{2}(t) }}+{\beta(t)\left(u(t)-{{\phi }}\left( \nu  \right) \right)} 
\end{aligned}
\end{equation}
where $\delta \left( t \right)=\left[ \begin{smallmatrix}
   1 & -1  \\
\end{smallmatrix} \right]{\mu(t)}$.
\begin{prop}
Consider the optimal control problem (P1) that is similar to (P), but with cost function $\Tilde{J}$, and the additional constraint \eqref{16}. Let $u^\star(t)$ be an optimal control and $\Tilde{\rho}^\star(t)$ the corresponding state trajectory response. Then, there exist the multiplier ${{\lambda }^{\star}}(t)$ that together with ${\delta }, \beta:\left[ {{t}_{0}},{{t}_{f}} \right]\to \mathbb{R}$ satisfy the PMP necessary conditions. 
More precisely, 
\begin{equation*}
    {\mathcal{H}}(\rho^\star,\lambda^\star,u^\star, \nu, \mu, \alpha, \beta ,t)\le{\mathcal{H}}(\rho^\star,\lambda^\star,u, \nu, \mu, \alpha, \beta ,t)
\end{equation*}
for all $t \in [t_0,t_f]$, and all feasible controls $u\in \Omega$. 
Moreover, 
\begin{equation*}
\begin{aligned}
& \frac{\partial {\mathcal{H}}}{\partial u}=(\left| \left. \lambda^\star\right\rangle\right\rangle - 2 \delta \left|\left.\Tilde{\rho}^\star \right\rangle\right\rangle)^{{T}} \tilde{\mathcal{L}}_u\left|\left. \tilde{\rho}^\star \right\rangle\right\rangle+2\eta {{u}^\star}+\beta=0\\
& \frac{\partial {\mathcal{H}}}{\partial \nu}=2\alpha\nu-\beta\frac{{{\partial\phi\left( \nu \right) }}}{\partial \nu}=0\\
\end{aligned}    
\end{equation*}
The remaining first-order necessary conditions for the state and costate variables are given as 
\begin{equation*}
\begin{aligned}
&\!\!{\left| \left. \dot { \tilde{\rho}} ^{\star} \right\rangle\right\rangle}=\frac{\partial {\mathcal{H}}}{\partial \left| \left. \lambda \right\rangle \right\rangle}=\Tilde{\mathcal{L}} \left|  {\Tilde{\rho}}^{\star}  \left. \right\rangle\right\rangle, \quad \left|  \left. \Tilde{\rho}^{\star} \left( {{t}_{0}} \right) \right\rangle \right\rangle = {{\rho }_{0}} \\
&\!\!\dot{ \left| \left. {{\lambda }^{\star}} \right\rangle\right\rangle}^T\!\!\!=\!\!\frac{-\partial {\mathcal{H}}}{\partial {\left| \left.  { \tilde{\rho }} \right\rangle\right\rangle}}\!=\!\Tilde{\mathcal{L}}^T \left| \left. \lambda^{\star}  \right\rangle\right\rangle -4\delta\Tilde{\mathcal{L}} \left| \left. \rho ^{\star}  \right\rangle\right\rangle, \left|\left.  \lambda^{\star} \left( t_f\right) \right\rangle\right\rangle = \lambda_f=\bf{0}\\
\end{aligned}
\end{equation*}
In addition, the transversality condition imposes that
\begin{equation*}
    {\mathcal{H}}\left( t_f \right)+\Gamma=0
\end{equation*}
\end{prop}
\begin{proof} For the indicated set of conditions, the extended Hamilton-Pontryagin function is 
\begin{equation*}
\begin{aligned}
&{\mathcal{H}}(\rho,\lambda,u, \nu, \mu, \alpha, \beta ,t)=(\left| \left. \lambda \left( {t} \right) \right\rangle\right\rangle - \mu^T(t) {{\nabla }_{\Tilde{\rho }}}h\left( \left|\left.\Tilde{\rho} (t)\right\rangle\right\rangle \right)^T\\
&\qquad\Tilde{\mathcal{L}}\left|{\tilde{ \rho} \left( {t} \right) }\left.\right\rangle\right\rangle
+\eta {{u}^{2}}\left( t \right)+\alpha{{\nu(t) }^{2}}+{\beta(t)\left(u(t)-{{\phi }}\left( \nu\right) \right)}\\ 
\end{aligned}
\end{equation*}
in which the state constraint $ h\left( \left|\left.\Tilde{\rho} (t)\right\rangle\right\rangle \right)$ and its multiplier $\mu(t)$ are expressed as
\begin{equation}\label{18}
 {h\left( \left|\left.\Tilde{\rho} (t)\right\rangle\right\rangle \right)}=
 \left( \begin{smallmatrix}
  \Tilde{P}[\chi(\rho)]-{{P}_{0}}  \\
   \alpha {{P}_{0}}-\Tilde{P}[\chi(\rho)] \\
\end{smallmatrix} \right)\le 0, \quad 
 \mu(t)=\left( \begin{smallmatrix}
   {{\mu }_{1}}\left( t \right) \\ {{\mu }_{2}}\left( t \right)  \\
\end{smallmatrix} \right)
\end{equation}
Therefore, ${{\nabla }_{\Tilde{\rho }}}{h\left( \left|\left.\Tilde{\rho} (t)\right\rangle\right\rangle \right)}=\left( \begin{matrix}
   {\nabla }_{\Tilde{\rho }}  \Tilde{P}[\chi(\rho)] &
  -{\nabla }_{\Tilde{\rho }}  \Tilde{P}[\chi(\rho)] \\
\end{matrix} \right)^T$, where $ {\nabla }_{\Tilde{\rho }}  \Tilde{P}[\chi(\rho)]=2\left|\left.\Tilde{\rho} (t)\right\rangle\right\rangle$. Then, the Pontryagin Hamiltonian forms as written in \eqref{20} from which the necessary conditions according to the PMP have been indicated.
\end{proof}
\begin{prop}
Let $\rho^\star(t)$, $u^\star(t)$, $\lambda^\star(t)$ be an optimal trajectory of (P1). Then, $\rho^\star(t)$, $u^\star(t)$, $\lambda^\star(t)$  is a local minimizer of the Pontryagin Hamiltonian \eqref{20} as long as $\alpha$ satisfies the condition 
\vspace*{-3mm}
\begin{equation}\label{eq:alpha}
\alpha>\frac{\beta}{2} \epsilon, 
\end{equation}
where $\epsilon = \max_{\nu \in \mathbb R} |\frac{{{\partial }^{2}}\phi(\nu) }{\partial {{\nu }^{2}}}|$ is a finite bound.
\end{prop}
\begin{proof}
Consider the control input  $\Bar{u}=\left(u,\nu\right)$. According to the second-order sufficient condition for local optimality, if the Pontryagin Hamiltonian has a positive definite Hessian with respect to $\Bar{u}$, then it is guaranteed that $u^\star(t)$ is the local minimizer of the Hessian, i.e., the generalized Legendre-Clebsch (L-C) condition guarantees that over a singular arc, the Pontryagin Hamiltonian is minimized. In this case, the Hessian is given by
\vspace*{-3mm}
\begin{equation*}\label{hessian}
     \mathcal{H}_{\Bar{u}\Bar{u}}=\left( \begin{smallmatrix}
   2\eta  & 0  \\
   0 & 2\alpha -\beta \frac{{{\partial }^{2}}\phi }{\partial {{\upsilon }^{2}}}  \\
\end{smallmatrix} \right)>0.
\end{equation*}
Clearly, it is positive definite since $\eta$ is positive, $\frac{{{\partial }^{2}}\phi }{\partial {{\upsilon }^{2}}}$ is bounded and $\alpha$ satisfies \eqref{eq:alpha}.
\end{proof}
\vspace{-0.2cm}
\section{Physics-Informed Neural Networks Based on the Theory of Functional Connections}
In this section, we exploit the newly developed Pontryagin PINN method derived from the theory of functional connection, \cite{d2021pontryagin, johnston2021theory}. We give a short overview on how physics-informed neural networks derived from TFC can be used to solve BVPs. Consider a generic differential equation that needs to be solved to obtain the dynamic vector $y(t)\in \mathbb{R}^n$, being the solution of the system, as follows
\vspace{-0.19cm}
\begin{equation*}
{F}\left( t,{{y}}\left( t \right),{{{\dot{y}}}}\left( t \right)\right)=0, \quad
{{{{y}}}}\left( t_k\right)={y}_{{t}_k},\quad k\in\{\emptyset, 1, 2, \ldots \}\cdot
\end{equation*}
\vspace{-0.03cm}
Note that the set in $k$ can be empty. Table \ref{Tab} describes (an adaptation of) the main steps of the application of TFC method developed in \cite{d2021pontryagin, johnston2021theory}. 
\begin{table}
 \hrulefill
\caption{Modelling and training the PoNN (main steps)}\label{alg1}
\begin{algorithmic}
\State 1. 
Morphing transformation
 $$\tau={{\tau}_{0}}+c\left( t-{{t}_{0}} \right)\leftrightarrow t={{t}_{0}}+\frac{1}{c}\left( \tau-{{\tau}_{0}} \right), \quad c>0$$
 \vspace{-0.2cm}
\State 2. Derive the approximated formulations for $y(t)$
\begin{equation}\label{27}   
\hat{y}(\tau)=g\left(\tau\right)+\sum\limits_{n=1}^{{k}}{\Omega_{k}\left( \tau \right){(y(\tau_k)-g(\tau_k))}}
\end{equation}
 \vspace{-0.2cm}
\State 3. Obtain a new 
equation as a function of the independent variable $\tau$, $g\left( \tau \right)$, and its derivative
\begin{equation}\label{ftild}
F(\tau ,\hat{y}(\tau ),\dot{\hat{y}}(\tau ))=\tilde{F}(\tau ,{g}(\tau ),\dot{{g}}(\tau ))=0
\end{equation}
 \vspace{-0.2cm}
\State 4. Formulate $g\left( \tau \right)$ through a single layer neural network 
 \begin{equation}
   {g}\left( \tau \right)=\sum\limits_{l=1}^{L}{{{\alpha }_{l}}\sigma\left( {{\omega }_{l}}z+{{b}_{l}} \right)}=\xi^{T}{h}\left( z \right)
\end{equation}
 \vspace{-0.2cm}
\State 5. Compute the derivatives of ${{g}}\left( \tau \right)$ as ${\dot{g}(\tau)}=\frac{{{d}}{{g}}}{d{{\tau}}}=\xi^{T}\frac{{{d}}{{h}}\left( \tau \right)}{d{{\tau}}}{{c}}$.
and obtain \eqref{ftild} in terms of unknowns, so $\tilde{\tilde{F}}(\tau,\xi)=0$.
\State 6. Discretize $\tau$ into $N$ points, and construct the loss matrix ${{\mathbb{L}}}_{n\times N}$ through the obtained set of differential equations as
\begin{equation} \label{29.6}
    {{\mathbb{L}}}\left( {{\xi }} \right)=\left\{ {{\tilde{\tilde{F}}}}\left( {{\tau}_{0}},{{\xi }} \right), {{\tilde{\tilde{F}}}}\left( {{\tau}_{1}},{{\xi }} \right),\ldots, {{\tilde{\tilde{F}}}}\left( {{\tau}_{N-1}},{{\xi }} \right),{{\tilde{\tilde{F}}}}\left( {{\tau}_{N}},{{\xi }} \right) \right\}
\end{equation}
\State 7. In order to obtain the unknown $\xi$,
compute the solution of $\mathbb{L}=0_{n\times{N}}$.
\end{algorithmic}
\label{Tab}
 \hrulefill
 \vspace{-0.5cm}
\end{table}
In \eqref{27}, the function $g\left( \tau \right): \mathbb{R}\to \mathbb{R}^n$ indicates a user-specified function, and $\Omega$ is the so-called switching function. 
The free function ${g}\left( \tau \right)$, as developed in \cite{huang2006extreme} based on the theory of extreme learning machine (ELM), can be modeled by a single hidden layer feedforward neural network, where the summation is over all $L$ hidden neurons, and $\sigma(\cdot)$ is the activation function. The output weight $\alpha_l$ and the input weights $\omega_l=[\omega_{1},\omega_{2},\cdots,\omega_{L}]$ connect the $l$th hidden node to the output node and input nodes, respectively. The bias of the $l$th hidden node is shown by $b_l$. 
Now, let proceed with the solution of the two point BVP resulted from the necessary optimality conditions for ($P_1$) that must be satisfied simultaneously. Taking into account the procedure indicated in Table \ref{Tab}, we approximate the state and costate as
\begin{equation*}
\begin{aligned}
  {\left|\left. { \tilde{\rho }} \right\rangle\right\rangle}\left( \tau,\xi_{\rho}  \right)&={{\left( {{\sigma}_{{ \tilde{\rho }} }}\left( \tau \right)-{{\Omega}_{1}}\left( \tau \right){{\sigma}_{{ \tilde{\rho }} }}\left( {{\tau}_{0}} \right)-{{\Omega }_{2}}\left( \tau \right){{\sigma}_{{ \tilde{\rho }} }}\left( {{\tau}_{f}} \right) \right)}^{T}}{{\xi }_{{ \tilde{\rho }} }}\\
  &\quad+{{\Omega}_{1}}\left( \tau \right){ {{\rho }_{0}}}+{{\Omega}_{2}}\left( \tau \right){{{\rho }_{f}}} \\ 
  \left|  \left. \lambda \right\rangle\right\rangle \left( \tau,\xi_{\lambda} \right) &= {{\left( {{\sigma}_{\lambda}}\left( \tau \right)-{{\Omega }_{2}}\left( \tau \right){{\sigma}_{\lambda}}\left( {{\tau}_{f}} \right) \right)}^{T}}{{\xi }_{\lambda}}+{{\Omega }_{2}}\left( \tau \right){ {{\lambda}_{f}}}\\
\end{aligned}
\end{equation*}
where $\Omega_1$ and $\Omega_2$ are switching functions expressed by, \cite{johnston2020fuel},
\begin{equation*}
{\small
\begin{aligned}
&{{\Omega }_{1}}\left( \tau \right)=1+\frac{2{{\left( \tau -{{\tau }_{0}} \right)}^{3}}}{{{\left( {{\tau }_{f}}-{{\tau }_{0}} \right)}^{3}}}-\frac{3{{\left( \tau -{{\tau }_{0}} \right)}^{2}}}{{{\left( {{\tau }_{f}}-{{\tau }_{0}} \right)}^{2}}}\\
&{{\Omega }_{2}}\left( \tau \right)=-\frac{2{{\left( \tau -{{\tau }_{0}} \right)}^{3}}}{{{\left( {{\tau }_{f}}-{{\tau }_{0}} \right)}^{3}}}+\frac{3{{\left( \tau -{{\tau }_{0}} \right)}^{2}}}{{{\left( {{\tau }_{f}}-{{\tau }_{0}} \right)}^{2}}}\\
\end{aligned}}
\end{equation*}
By taking the derivatives of the approximated expressions of state and costate equation we obtain
\begin{equation*}
\begin{aligned}
&{\left|\left. \dot { \tilde{\rho }} \right\rangle\right\rangle}\left( \tau ,\xi_{\rho}  \right)\!=\!c{{\left( {{\sigma}^{\prime}_{\Tilde{\rho}}}\left( \tau  \right)-{{\Omega }_{1}^{\prime}}\left( \tau  \right){{\sigma}_{\Tilde{\rho}}}\left( {{\tau }_{0}} \right)-{{\Omega }_{2}^{\prime}}\left( \tau  \right){{\sigma}_{\Tilde{\rho}}}\left( {{\tau }_{f}} \right) \right)}^{T}}\!{{\xi }_{\Tilde{\rho}}}\\
&\quad+{{\Omega }_{1}^{\prime}}\left( \tau  \right){{{\rho }_{0}}}+{{\Omega }_{2}^{\prime}}\left( \tau  \right){{{\rho }_{f}} } \\ 
&\dot{ \left| \left. {{\lambda }} \right\rangle\right\rangle}  \left( \tau ,\xi_{\lambda} \right) \! = \!  c{{\left( {{\sigma}^{\prime}_{\lambda}}\left( \tau \right)-{{\Omega }_{2}^{\prime}}\left( \tau  \right){{\sigma}_{\lambda}}\left( {{\tau }_{f}} \right) \right)}^{T}}{{\xi }_{\lambda}}+{{\Omega}_{2}^{\prime}}\left( \tau \right){{{\lambda}_{f}} }\\
\end{aligned}
\end{equation*}
Regarding the control variables, the following functional approximation are introduced
\begin{equation*}
 u\left( \tau,\xi  \right)=\sigma_{u}^{T}\left( \tau \right){{\xi }_{u}}, \quad \nu\left(\tau,\xi \right)=\sigma_{\nu}^{T}\left( \tau \right){{\xi }_{\nu}}
\end{equation*}
We also expand the equality constraint multiplier $ \beta\left( \tau,\xi \right)=\sigma_{\beta}^{T}\left( \tau \right){{\xi }_{\beta}}$.
In addition, the state constraint multipliers are 
\begin{equation*}
 \mu_1\left(\tau,\xi  \right)=\sigma_{\mu_1}^{T}\left(\tau \right){{\xi }_{\mu_1}}, \quad  \mu_2\left(\tau,\xi \right)=\sigma_{\mu_2}^{T}\left( \tau \right){{\xi }_{\mu_2}} \\ 
\end{equation*}
In line with the ELM algorithm, \cite{huang2006extreme}, for the free final time problems, the vector of PoNNs' parameters to be learned is constructed as $    \xi =\left\{ \begin{smallmatrix}
   {{\xi }_{\Tilde{\rho} }} & {{\xi }_{\lambda }} & {{\xi }_{u}} & {{\xi }_{\nu}} & {{\xi }_{\beta}}& {{\xi }_{\mu_1}}& {{\xi }_{\mu_2}} & {c}  \\
\end{smallmatrix} \right\}^T$.
Now, we express the set of loss functions to be minimized as
\begin{equation*}
\begin{aligned}
&\mathbb{L}_{\Tilde{\rho}}={\left| \left. \dot { \tilde{\rho }} \right\rangle\right\rangle}-\Tilde{\mathcal{L}} \left|  {\Tilde{\rho}}  \left. \right\rangle\right\rangle\\
&\mathbb{L}_{\lambda}= \dot {\left| \left.{ \lambda} \right\rangle\right\rangle}^T-(\Tilde{\mathcal{L}}^T \left| \left. \lambda  \right\rangle\right\rangle -4\delta\Tilde{\mathcal{L}} \left| \left. \Tilde{\rho} \right\rangle\right\rangle)\\
&\mathbb{L}_u=(\left| \left. \lambda \right\rangle\right\rangle - 2 \delta \left|\left.\Tilde{\rho} \right\rangle\right\rangle)^{{T}} \tilde{\mathcal{L}}_u\left|\left. \tilde{\rho} \right\rangle\right\rangle+2\eta {{u}}+\beta\\
&\mathbb{L}_\nu=2\alpha\nu-\beta{{\phi }^\prime}\left( \nu \right), \quad 
\mathbb{L}_\phi=u-\phi(\nu)\\
&\mathbb{L}_{{\mathcal{H}}}={\mathcal{H}}\left( t_f \right)+\Gamma\\
&\mathbb{L}_{\mu_1}=\Tilde{P}[\chi(\rho)]-{{P}_{0}},\quad 
\mathbb{L}_{\mu_2}=\alpha {{P}_{0}}-\Tilde{P}[\chi(\rho)]  \\
\end{aligned}
\end{equation*}
leading to an augmented form of the loss function \eqref{29.6} as
\begin{equation}\label{51}
    \mathbb{L}= \left\{ \begin{matrix}
   \mathbb{L}_{\Tilde{\rho}} &  \mathbb{L}_{\lambda} & \mathbb{L}_u & \mathbb{L}_\nu & \mathbb{L}_\phi & \mathbb{L}_{{\mathcal{H}}} & \mathbb{L}_{\mu_1} & \mathbb{L}_{\mu_2}  \\
\end{matrix} \right\}
\end{equation}
Equation \eqref{51} can be solved by a numerical minimization scheme, and the PoNN's parameter will be learnt during the procedure. The iterative least-square method has proved to be an efficient scheme for such problem, \cite{d2021pontryagin}. Through this method, the estimation of $\xi$ is adjusted and refined to improve accuracy and convergence towards the desired outcome at $k+1$th iteration, such that ${{\xi }_{k+1}}={{\xi }_{k}}+\Delta {{\xi }_{k}}$, where $\Delta {{\xi }_{k}}=-{{\left( \mathbb{J}{{\left( {{\xi }_{k}} \right)}^{T}}\mathbb{J}\left( {{\xi }_{k}} \right) \right)}^{-1}}\mathbb{J}{{\left( {{\xi }_{k}} \right)}^{T}}\mathbb{L}({{\xi }_{k}})$, in which $\mathbb{J}$ is the Jacobian matrix, compiling the partial derivatives of the loss function with respect to each unknown parameter, and, therefore provides a complete representation of the sensitivity of the losses to changes in the unknowns. The iterative procedure continues to be repeated until the convergence criteria is satisfied, meaning that for a predefined tolerance we reach to ${{L}_{2}}\left[ \mathbb{L}\left( {{\xi }_{k}} \right) \right]<\varepsilon$.
\section{Simulation Results}
In the following, we show the feasibility of our results in a numerical study. Let consider a quantum system consisting of a two-level atom and a vacuum environment. The dissipation is captured via the decay of the atom through a photon emission, which is a result of the atom - environment interaction. The environment in this case is the surrounding vacuum state. Therefore, the system dynamics must be expressed through the master equation \eqref{1}, in which the term $H(u(t))$ describes the total atom-vacuum quantum-mechanical Hamiltonian as\begin{equation}
    H\left( u\left( t \right) \right)=E\left| 1 \right\rangle \left\langle  1 \right|+u\left( t \right)\left( \left| 0 \right\rangle \left\langle  1 \right|+\left| 1 \right\rangle \left\langle  0 \right| \right)
\end{equation}
where $u(t)$ is the driving control coherently switching between the two states. 
We work with the implemented dynamics in \eqref{8}. For this problem, the damping rate $\gamma$ is the atom-vacuum coupling, and the Liouvillian superoperator in the Fock-Liouvillian space is
\begin{equation}
    {\mathcal{L}}=\left( \begin{smallmatrix}
   0 & iu & -iu & \gamma   \\
   iu & -iE-\frac{\gamma }{2} & 0 & -iu  \\
   -iu & 0 & -iE-\frac{\gamma }{2} & iu  \\
   0 & -iu & iu & -\gamma   \\
\end{smallmatrix} \right)
\end{equation}
which is mapped to the extended superoperator $\tilde{\mathcal{L}}$.
\begin{figure}[t!]
  \centering
  \includegraphics[scale=0.62]{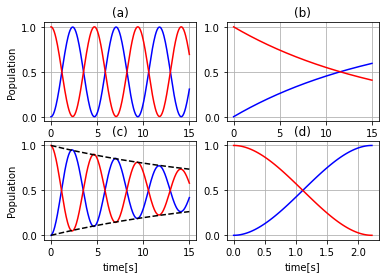}
  \caption{Population evolution under quantum dynamics. The blue and red lines represent the evolution of $\rho_{00}$ and $\rho_{11}$, respectively. (a) Coherent driving with no decay. (b) Decay with no coherent driving. (c) Coherent driving and decay. (d) Optimal coherent driving and decay. }
  \label{subs}
  \vspace{-0.5cm}
\end{figure}
We represent the results of \eqref{8} to show the element wise time behaviour of the density operator. For all results, the state is initiated as $\rho_{11}=1$, with no coherence between different states. First, we neglect the dissipation part, so the solution reduces to the resolution of the Liouville-von Neumann equation, see Fig \ref{subs}(a). As the next case, we consider the dissipation while there is no coherent driving. The population of the excited state experiences an exponential decay, see Fig \ref{subs}(b). Afterward, it becomes feasible to compute the behavior of a dissipative quantum system that undergoes both coherent driving and decay. In such cases, oscillations and decay co-occur as both behaviors exist simultaneously, see Fig \ref{subs}(c). In the following, we plot the state evolution under the action of optimal control. We initiate by the pure state $\rho_{11}=1$ and target to $\rho_{00}=1$, while following the constraints on purity preservation. The population evolution is shown in Fig \ref{subs}(d). According to the data depicted in the graph, there is a clear trend of exponential decay in the population of the initial state, while there is a corresponding upward trend in the population of the target as time passes. In order to check the security level of reaching the target, we have to calculate the transition probability known as the quantum fidelity. The term fidelity refers to the degree of similarity between two quantum states, typically the system state $\rho(t)$ and target $\sigma$. We study fidelity in terms of state purity computed by
$\mathcal{F}\left( \rho ,\sigma  \right)=\frac{{{\left( tr\left( \rho \sigma  \right) \right)}^{2}}}{P\left( \rho  \right)P\left( \sigma  \right)}$, \cite{indrajith2022fidelity}. 
Since our target $\sigma$ is a pure state, i.e., $\sigma =\left| \psi  \right\rangle \left\langle  \psi  \right|$, fidelity can simply be assessed as $   \mathcal{F}\left( \rho ,\sigma  \right)=\frac{\left\langle  \psi  \right|\rho \left| \psi  \right\rangle }{P\left( \rho  \right)}$.
Figure \ref{fig6} shows the effects of purity preservation on fidelity.
\section{CONCLUSIONS}
In this study, we proposed a framework for preserving quantum purity during a quantum state transition problem. Specifically, we aim to minimize both the time and energy required for the transition, while adhering to the Lindblad master equation, which governs the system dynamics. To achieve this goal, we employed a combination of two techniques, namely the Gamkrelidze revisited method and the concept of saturation functions and system extensions. The resulting boundary value problem is then solved using Pontryagin neural networks, which are well-suited for this type of problem formulation. Our approach allows us to preserve the purity of the quantum state during the transition while minimizing the resources required. Furthermore, we analyze the effects of state constraints on the evolution of quantum fidelity, providing a numerical example for a two-level system. As future work, we intend to extend our approach to higher-order dimensional systems, which could yield valuable insights and applications in various fields, including quantum computing and quantum communication. Our results have important implications for the practical implementation of quantum state transitions.
\begin{figure}[t!]
  \centering
  \includegraphics[scale=0.4]{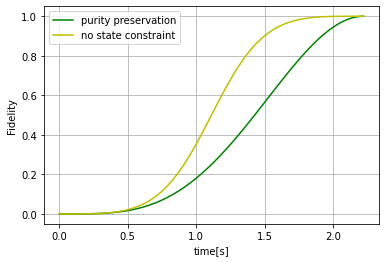}
  \caption{State transition probability quantified by purity}
  \label{fig6}
 \vspace{-0.65cm} 
\end{figure}

\end{document}